\documentclass[sigconf, nonacm]{acmart}
\usepackage{subcaption}
\usepackage{balance}
\usepackage{listings}

\usepackage{fancyhdr}
\usepackage{fancyhdr}
\usepackage{footmisc}

\AtBeginDocument{%
  }

\begin{document}

\title{Quantum Modeling of Spatial Contiguity Constraints}

\author{Yunhan Chang$^{1,2}$}
 \affiliation{%
   \department{$^1$Department of Computer Science and Engineering}
   \position{$^2$Center for Geospatial Sciences}
   \institution{University of California, Riverside}\country{Riverside, CA, USA}
 }
 \email{ychan268@ucr.edu}

 \author{Amr Magdy$^{1,2}$}
 \affiliation{%
   \department{$^1$Department of Computer Science and Engineering}
   \position{$^2$Center for Geospatial Sciences}
   \institution{University of California, Riverside}\country{Riverside, CA, USA}
 }
 \email{amr@cs.ucr.edu}

 \author{Federico M. Spedalieri$^{3,4}$}
 \affiliation{
  \department{$^3$Department of Electrical and Computer Engineering}
  \position{$^4$Information Sciences Institute}
  \institution{University of Southern California}\country{Los Angeles, CA, USA}
 }
 \email{fspedali@isi.edu}


\begin{abstract}
Quantum computing has demonstrated potential for solving complex optimization problems; however, its application to spatial regionalization remains underexplored.
Spatial contiguity, a fundamental constraint requiring spatial entities to form connected components, significantly increases the complexity of regionalization problems, which are typically challenging for quantum modeling. This paper proposes novel quantum formulations based on a flow model that enforces spatial contiguity constraints. Our scale-aware approach employs a Discrete Quadratic Model (DQM), solvable directly on quantum annealing hardware for small-scale datasets. In addition, it designs a hybrid quantum-classical approach to manage larger-scale problems within existing hardware limitations. This work establishes a foundational framework for integrating quantum methods into practical spatial optimization tasks.
\end{abstract}

\maketitle

\setlength{\footskip}{45pt} 

\thispagestyle{fancy}
\fancyfoot[C]{} 
\fancyfoot[R]{This paper has been revised and accepted at the Q-Data 2025 Workshop, in conjunction with ACM SIGMOD 2025.}

\section{Introduction}

\begin{figure}[t]  
    \centering
\includegraphics[width=0.8\linewidth]{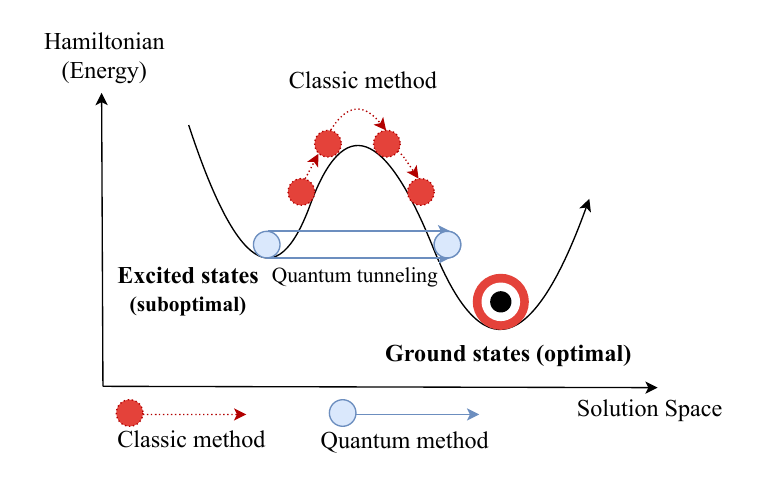} 
    \caption{Quantum Annealing Process vs. Classic Method}
    \Description{A figure illustrating the process of quantum annealing, showing potential energy landscapes and quantum advantage over classic method.}
    \label{fig:qa_process}
\end{figure}

Quantum annealing~\cite{kadowaki1998quantum, Santoro2006QA, Finnila1994, Viladomat_Jasso_2023, Kumar2018QC} is a computational technique inspired by quantum mechanics, designed to solve complex optimization problems by exploring a vast solution space through quantum fluctuations. Unlike classical simulated annealing~\cite{Kirkpatrick1983, Koshka2020DWave}, which is based on thermal fluctuations, quantum annealing uses quantum tunneling to escape the local optimum and potentially converge on the global optimum. This process involves initializing a quantum system in a superposition state that fluctuates between 0 and 1, then reducing quantum effects until the system settles into a low-energy configuration, called a ground state, representing a solution.
Figure~\ref{fig:qa_process} shows a solution space for an arbitrary NP-hard optimization problem represented by an energy curve, called the Hamiltonian.
Every possible solution of the NP-hard problem is a point on that curve, and its solution quality corresponds to the overall energy of the underlying quantum system.
The Hamiltonian~\cite{RevModPhys.80.1061} comprises multiple local maxima (hills) and minima (valleys), each corresponding to a different candidate solution for that NP-hard problem.
In the solution space, a transverse field functions as a subtle fluctuating perturbation that allows the system to quantum tunnel like a marble through energy barriers rather than surmounting them via classical thermal activation~\cite{Kirkpatrick1983}.
In Figure~\ref{fig:qa_process}, reaching from a suboptimal solution of the NP-hard problem (an excited state represented by the first valley) to the optimal solution (a ground state represented by the second valley), a classical probabilistic method must surmount the hill and sample many states in between.
In contrast, the quantum annealer can approach the optimal solution while exploring a significantly smaller number of states through quantum tunneling.
Quantum annealing can be applied to a broad range of NP-hard problems that traditionally depend on classical heuristic methods, including spatial regionalization that is recognized as a spatial query in the database community~\cite{Liu2022PRUC,Kang2022EMP}.
Notably, several companies have developed hardware that implements a quantum annealer aimed at addressing a variety of optimization challenges~\cite{web:annealercompanies}.

In recent years, significant quantum hardware advancements have been made.
To the best of our knowledge, the most advanced quantum annealer is developed and maintained by \textit{D-Wave Quantum Inc.}, and it offers approximately \textbf{5000 qubits}~\cite{King2023Quantum}.
However, current systems remain constrained by notable limitations~\cite{King2023Quantum,quantum7010002} that render existing quantum computers primarily suitable for low-scale problems.
To address this challenge, D-Wave has developed hybrid solvers that integrate classical CPUs with Quantum Processing Units (QPUs) at different phases.
The classical computer handles problem decomposition, setting initial conditions, and refining results using well-established optimization methods, while the QPU leverages quantum annealing. These solvers, such as DQM~\cite{dwave_dqm}, CQM~\cite{dwave2022CQM}, and NL solvers~\cite{Osaba2024DWave}, can partition large-scale problems into smaller sub-problems, thereby allowing QPUs to solve them efficiently. These solvers are designed for constrained optimization problems that aim to minimize an arbitrary objective function.

Quantum annealing is applied in many optimization problems.
Clustering is a classic optimization problem that groups objects or vectors in d-dimensional space according to certain criteria.
Spatial regionalization, i.e., finding spatial regions, is a refined clustering technique that groups spatial polygons, rather than points, of small areas to compose larger regions.
Regionalization has been widely adopted across diverse domains~\cite{Yunfan2024Pyneapple-R, AlrashidLM23,Alrashid2022SMP, Kang2022EMP}, such as economics, e.g., to address imbalances in economic development~\cite{church2020p}, and urban planning, e.g., for resource allocation in urban construction~\cite{Guo2008}. 
Spatial regions aggregated by areas must be explicitly spatially contiguous, i.e., all areas within the same region form one connected component in space.
The regionalization problem is NP-hard~\cite{Liu2022PRUC}, and solving it optimally is intractable due to contiguity and combinatorial constraints~\cite{AlrashidLM23, Alrashid2022SMP,Kang2022EMP}.
Currently, classical heuristics still take long hours of runtime to produce a suboptimal objective value~\cite{AlrashidLM23,Liu2022PRUC,Kang2022EMP}.
Quantum annealers, with their inherent ability to tunnel through local minima, offer a new computational method that may be more effective in handling these complexities.
The first step toward achieving a quantum advantage in spatially driven problems is ensuring spatial contiguity among neighboring spatial entities.
While previous research has predominantly focused on enhancing the efficiency of classical clustering problems through quantum computing methods~\cite{Viladomat_Jasso_2023, Gemeinhardt2021}, the integration of spatial contiguity constraints remains an ignored topic.

In this paper, we present the first quantum model that ensures spatial contiguity constraints among neighboring spatial entities.
Without loss of generality, we use the p-region problem as an example that requires spatial contiguity.
However, our model is general and can be used to ensure spatial contiguity in any spatial application and among any set of spatial entities, such as electoral redistricting~\cite{Validi2022Political}, forest planning~\cite{Carvajal2013Forest}, site selections~\cite{Murray2023Land}, wildlife conservation~\cite{Murray2023Wildfire}, and habitat corridor design~\cite{StJohn2018Corridor}.

Ensuring spatial contiguity poses challenges due to the limitations of current quantum hardware: storing neighborhood relationships among spatial entities can quickly consume a large number of qubits to represent variables, thereby complicating the design of efficient quantum solutions.
A critical limitation arises when larger datasets comprise more spatial entities, such as area polygons in regionalization, leading to a more significant number of combinations of assignments that require additional variables, which poses a challenge to current QPUs with a maximum of 5000 qubits.

To address the scalability challenge, we propose two models that accommodate datasets of varying scales.
For small-scale datasets, we use a \textit{flow-based model} that uses neighborhood flow variables to enforce spatial contiguity, formulated using the Discrete Quadratic Model (DQM)~\cite{dwave_dqm}.
Our model reduces the number of variables by exploiting the low degree of spatial connectivity between real entities.
To handle larger datasets, we propose a hybrid approach that blends classical and quantum computations when it is infeasible to process the entire dataset directly on QPUs.
The classic computing is used for operations that neither make use of quantum optimization nor represent a performance bottleneck.
Complementarily, quantum phase primarily adjusts the region boundaries to explore alternative solutions, and the process is iteratively integrated to achieve optimal solution quality.

\section{Related Work}


\noindent \textbf{Spatial contiguity in classical regionalization}.
All variations of spatial regionalization problems require spatial contiguity between spatial entities in the same region~\cite{Aydin02112021,wei2021efficient,duque2011pregions, AlrashidLM23, Alrashid2022SMP,rae2011geography,duque2012maxp, Kang2022EMP, Yunfan2024Pyneapple-R, Liu2022PRUC}.
In classical computing, spatial contiguity is typically enforced through simple geometric operations that check geometric intersections, while it is seldom expressed as optimization constraints in mathematical models.
The latter is more expressive and easier to adapt from a quantum optimization perspective.
Various mathematical formulations have emerged to target different optimization objectives.
For example, both p-regions~\cite{duque2011pregions} and p-compact regions problems~\cite{rae2011geography} partition the space into p regions, but the former minimizes the inter-regional heterogeneity while the latter maximizes the spatial compactness of each region.

Former methods have been proposed to ensure spatial contiguity in regionalization, but none of them works well for quantum computers.
Duque et al. model a spatial regionalization problem~\cite{duque2011pregions} as a mixed integer programming (MIP) formulation by representing polygons and their attributes with distinct variables. However, MIP approaches for regionalization problems are expensive because they must consider every possible assignment of variables to achieve an exact solution for an NP-hard problem, leading to a rapid growth in the number of required variables.
This quickly becomes infeasible for larger datasets. Combined with the current limitations on the number of available qubits, MIP approaches are unsuitable for quantum processing units at scale.
Therefore, it is essential to carefully design both decision variables and auxiliary variables to reduce the problem's dimensionality.
This reduction is critical for developing a Discrete Quadratic Model (DQM)~\cite{dwave_dqm} on D-Wave or a Quadratic Unconstrained Binary Optimization (QUBO) formulation that is computationally tractable. 


\textbf{Quantum-based clustering.}
Regionalization is a clustering problem that groups spatial polygons.
The quantum computing literature has many research efforts in clustering problems among other machine learning operations~\cite{Zaman:2023sym}, such as quantum kNN clustering~\cite{Viladomat_Jasso_2023} approaches that encode medoids selection into a QUBO model~\cite{bauckhage2019qubo}, and quantum community detection~\cite{Muhuri2024Quantum} methods that optimize network modularity via an Ising formulation. Although these works demonstrate promise for point-based clustering tasks, they do not enforce additional complex constraints: spatial contiguity, which significantly makes regionalization more complicated than traditional clustering.

\vspace{4pt}

\begin{figure*}[t]
  \centering
  \begin{subfigure}{0.35\textwidth}
    \centering
    \includegraphics[width=\linewidth]{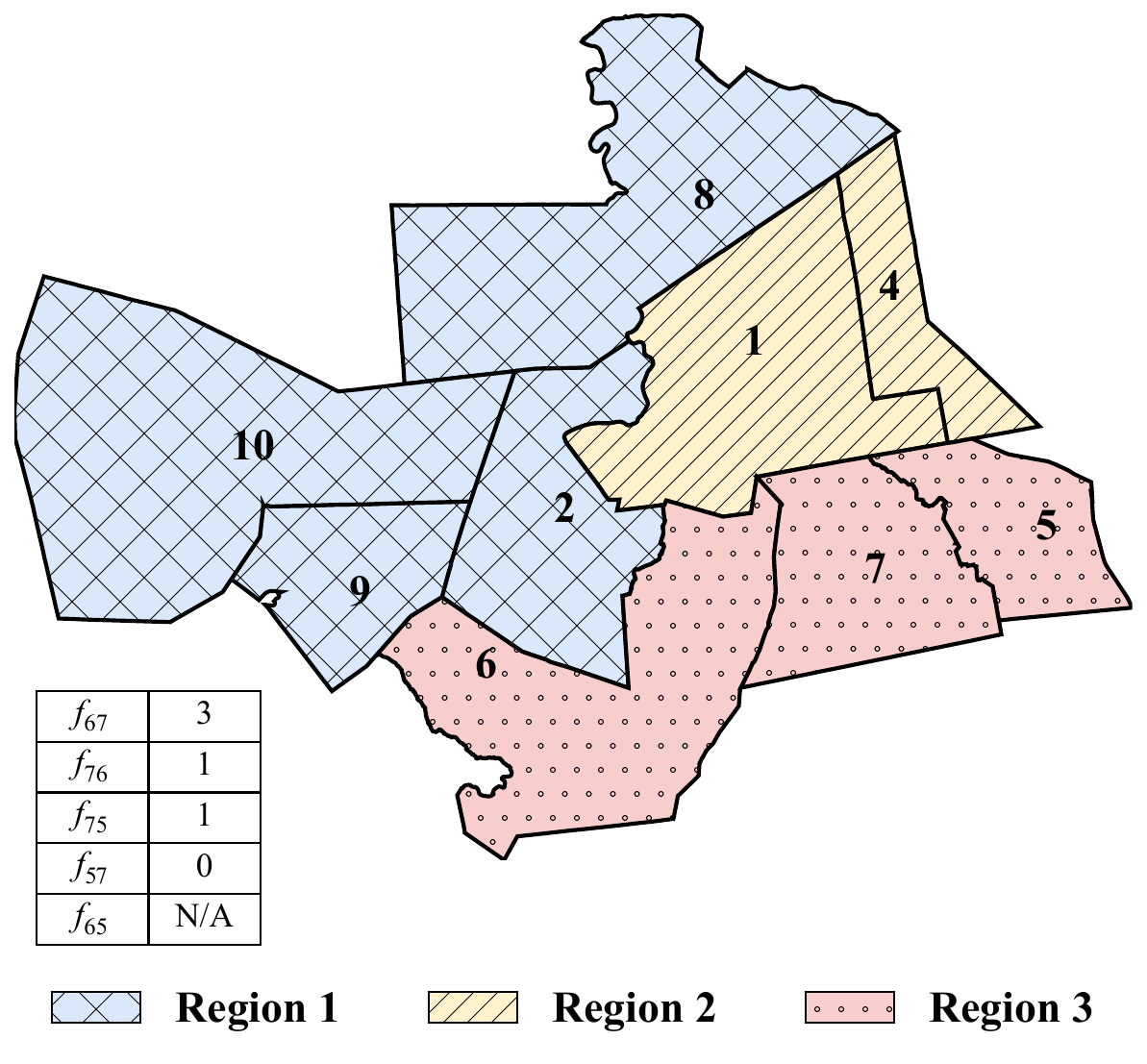}
    \caption{Connected Region 1 and Region 3}
    \label{fig:connected}
  \end{subfigure}
  \hspace{0.02\textwidth}
  \begin{subfigure}{0.35\textwidth}
    \centering
    \includegraphics[width=\linewidth]{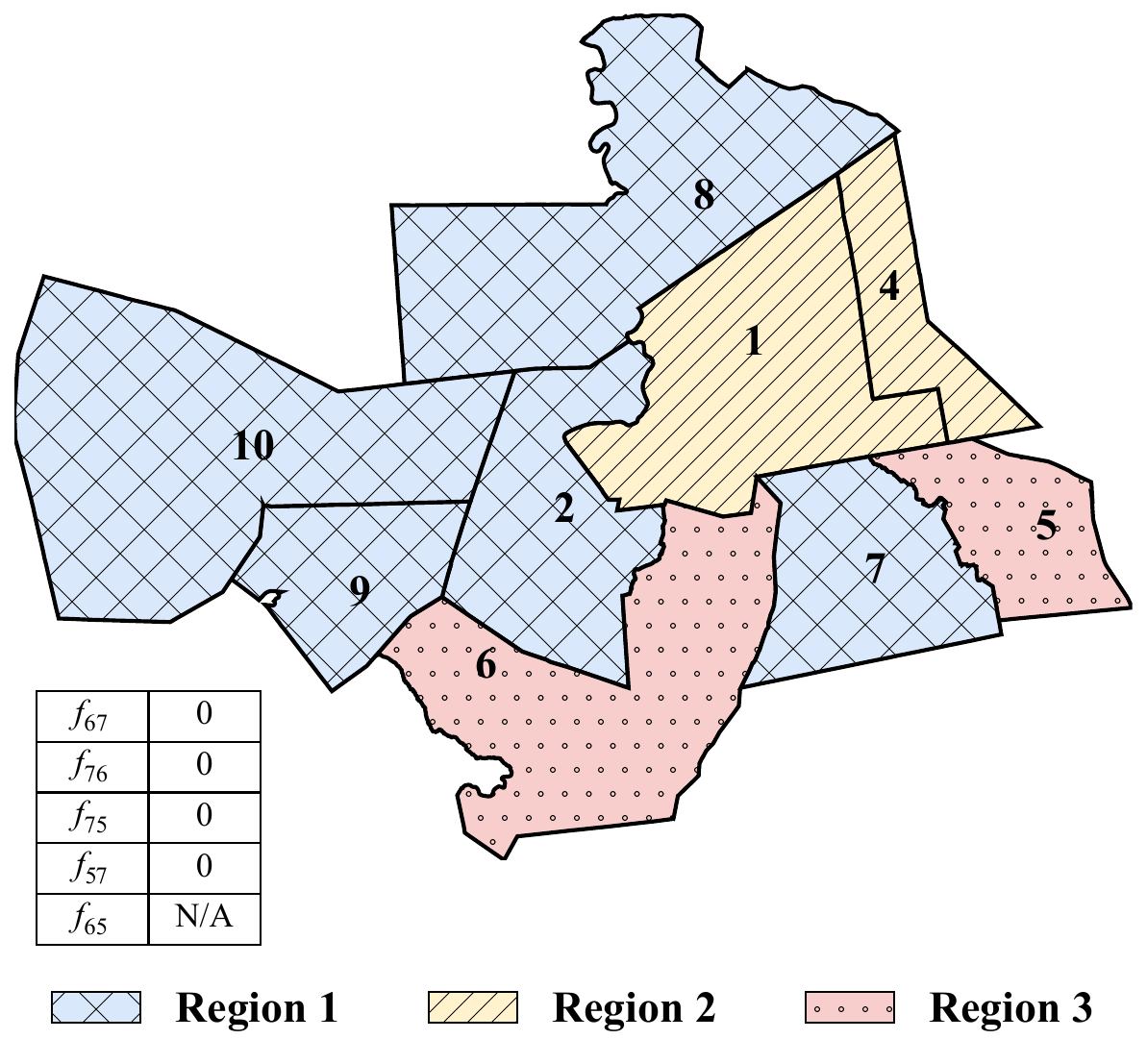}
    \caption{Disconnected Region 1 and Region 3}
    \label{fig:disconnected}
  \end{subfigure}
  \caption{Example of Connected and Disconnected Regions}
  \label{fig:figure2}
\end{figure*}

\vspace{-12pt}
\section{Background and Problem Formulation}



\textbf{Quantum Modeling}.
Quadratic Unconstrained Binary Optimization (QUBO)~\cite{Glover2022QUBO} provides a mathematical framework where problems are expressed as the minimization of a quadratic function over binary variables. In a QUBO formulation, the objective is represented as a quadratic polynomial, allowing a wide range of combinatorial optimization problems to be reformulated into this standard model. Its simplicity and versatility have made QUBO a favored approach for mapping complex problems, such as scheduling, clustering, and portfolio optimization, onto specialized hardware like quantum annealers. 
QUBO models are compatible with both classical and quantum solvers~\cite{Callison2022Hybrid}, making them powerful to empower hybrid solvers for a variety of problems.

Discrete Quadratic Models (DQM)~\cite{dwave_dqm} extend the QUBO framework by allowing variables to take on multiple discrete values rather than being restricted to binary values. This generalization is particularly useful for problems that naturally involve multi-valued decisions, providing a richer and more flexible modeling approach.
DQM formulations are directly supported on D-Wave’s quantum hardware. By accommodating a broader range of constraints and variable domains, DQM enables the modeling of intricate systems in areas such as scheduling problems~\cite{Schworm2023} and community detection~\cite{Wierzbinski2023}.
This reflects the ongoing evolution of quantum computing methods, as researchers seek to solve real-world problems that extend beyond the binary scope of traditional QUBO models.


\vspace{4pt}
\noindent \textbf{Problem Formulation.} Without loss of generality, we use spatial regionalization, specifically the p-regions problem, as a prominent example of a spatial problem that requires spatial contiguity among spatial entities.
All variations of spatial regionalization require spatial contiguity~\cite{duque2011pregions, duque2012maxp, Alrashid2022SMP, Kang2022EMP, Liu2022PRUC}, and the p-regions problem~\cite{duque2011pregions} is a simple representative variation that focuses on this specific requirement, which is the main focus of this paper.
The p-regions problem is defined as follows:


\vspace{-4pt}
\begin{definition}[p-Regions Problem]
\label{def:pregions}
Given:
(1)~A set of $n$ spatial areas: $A = \{ a_{1}, a_{2}, ..., a_{n} \}$.
(2)~An integer $p$.
p-regions finds a partition of regions $R =  \{r_{1}, r_{2}, ..., r_{p}\}$ of size $p$, where each region $r_i$ is a non-empty set of spatially contiguous areas, $|r_i| \geq 1$, so that:
(i)~$r_i \cap r_j = \Phi, \forall r_i,r_j \in R \wedge \; i \neq j$, i.e., all regions are disjoint.
(ii)~$\bigcup_{i=1}^{p} r_i = A$.

(iii)~The heterogeneity of $R$ is measured by the total pairwise dissimilarity among all areas within $R$:
\[
\text{Hetero}(R) = \sum_{\substack{i, j \in R \\ i < j}} w_{ij},
\]
where $w_{ij}$ is the dissimilarity between areas $i$ and $j$. The objective is to minimize $\text{Hetero}(R)$.

\end{definition}

\vspace{-12pt}
\section{Flow-based Spatial Contiguity}

This section describes the high-level mechanism that inspires our quantum model.
Duque et. al.~\cite{duque2011pregions} model spatial contiguity using three different models.
Our quantum model is inspired by the flow-based mechanism, which extends the original Shirabe model~\cite{shirabe2005model}, as it performs the best empirically among the three models.
The flow model has the fewest constraints, so it scales up better given the limitations of existing quantum hardware.
In addition, we adapt techniques that reduce the number of variables to optimize the model for the quantum hardware.
Specifically, we pick one root for each region and limit flow variables only to neighboring spatial entities.
As the spatial connectivity between real entities is typically small, this significantly reduces the number of flow variables needed.
Without loss of generality, and for simplicity, we consider spatial entities as spatial areas, e.g., city blocks, counties, or states.

Shirabe's spatial unit allocation model~\cite{shirabe2005model} utilizes an embedded network flow approach.
Spatial areas are modeled as a network, forming a graph where spatial areas are represented as nodes and spatial neighborhood relationships are represented as edges.
So, two areas are connected only if they share common boundaries in space.
Shirabe's model designates a sink node and ensures that all allocated areas remain connected via flow paths directed toward the sink.
Compared to Shirabe's single-region model, our multi-region flow model extends the network flow concept by enabling multiple distinct networks, each corresponding to a separate region, thus extending the utility in more regionalization problems. For each region, a flow network is established to check the connectivity of the region.

The flow idea is simple.
Flow is conceptually similar to water moving through pipes between reservoirs, where the inflow and outflow of each area must satisfy specific conservation requirements.
Every flow network of a region designates a source area that emits flow throughout the region through the flow links, i.e., edges that connect neighboring areas.
When an area is assigned to that region, it must receive flow from the region's source.
If there is no flow path from the region's source, then this area is disconnected.
These flow paths are exclusively within the region and do not extend beyond its borders.







Figure~\ref{fig:figure2} shows an example for three regions.
In Figure~\ref{fig:connected}, Regions~1 and~3 are connected, while in Figure~\ref{fig:disconnected} they are disconnected by assigning area~7 to Region 1.
Focusing on Region~3 in both cases, we designate area~6 as Region~3's source.
When area~6 emits flow through the flow network in Figure~\ref{fig:connected}, this flow goes to area~7 and, in turn, to area~5.
So, every area in Region~3 receives a positive inflow, meaning it is a connected region.
Note that flow only happens between direct neighbors, so area~6 cannot emit flow directly to area~5, but it can emit it indirectly through area~7.
On the contrary, in Figure~\ref{fig:disconnected}, the flow does not go from area~6 to area~7 as it does not belong to Region~3 in this case.
As flow transfers only through neighboring links, area~5 does not receive any flow from area~7. So, the flow of Region~3 source (area~6) does not reach area~5, and it is marked disconnected.
This idea is modeled using discrete variables, as detailed in the following section, with the DQM model.

\section{DQM Formulation}

Spatial contiguity is a complex constraint and challenging to express in a mathematical model to fit a quantum computer.
We exploit recent developments on quantum annealer and hybrid quantum-classical systems such as: (1)~solvers for \emph{Discrete Quadratic Models} (DQM)~\cite{dwave_dqm}, a generalization of QUBO~\cite{Glover2022QUBO} that supports multi-label, discrete decision variables, and (2)~hybrid solvers, such as \emph{Constrained Quadratic Models (CQM)}~\cite{dwave2022CQM} and \emph{Nonlinear Hybrid Solvers} (NL-solvers)~\cite{Osaba2024DWave}, that are capable of incorporating integer, linear, and non-linear constraints alongside quadratic objectives. 
These solvers enable a broader range of real-world optimization tasks to be mapped onto a hybrid quantum-classical architecture~\cite{Callison2022Hybrid}.

To enforce spatial contiguity on a D-wave quantum computer, we encode the flow model in DQM.
Assuming we have $n$ areas and $p$ regions, we build DQM by representing each area as a discrete variable with $p$ states, formulating a quadratic objective function, and incorporating a penalty term to achieve spatial contiguity using flow contiguity constraints. 
Let \[I=\{1,2,\dots,n\}\] be the set of areas. For each area \(i\in I\) we introduce the assignment variable
\[d_i\in\{1,2,\dots,p\},\]
so that \(d_i=k\) means that area \(i\) is assigned to region \(k\), \(1 \leq k \leq p\).

Let \(E\) denote the set of directed edges \((i,j)\) representing neighboring (contiguous) pairs; that is \(N(i)\) is a set of \(i\)'s neighbors, \(j\in N(i)\) if areas \(i\) and \(j\) share a common border.

\subsection*{Auxiliary Variables}

\begin{itemize}
    \item Introduce a discrete variable
    \[
    f_{ij}\ge 0,
    \]
    representing flow from area \(i\) to area \(j\), only if \((i,j) \in E\), i.e., \(i\) and \(j\) are neighbors. This significantly reduces the number of flow variables due to the low degree of connectivity between real spatial entities.
    
    \item Define an auxiliary variable \(u_{ij}\ge 0\) to examine whether two areas belong to the same region, \(u_{ij}=|d_i-d_j|\); hence, if \(d_i=d_j\) then \(u_{ij}=0\), meaning area \(i\) and area \(j\) belong to the same region. \(u_{ij}\) could help us determine the area relation between and inside regions.
\end{itemize}

\subsection*{Intra-Region Flow Constraint}

To ensure that flow is allowed only between areas assigned to the same region, we require that for each \((i,j)\in E\):
\[
f_{ij}\le M\,(1 - u_{ij}),
\]
where \(M\) is a suitably large constant. It worth noting that if \(d_i\neq d_j\) then \(u_{ij}\ge 1\) so that \(1-u_{ij}\le 0\) and thus \(f_{ij}=0\).

\vspace{-4pt}
\subsection*{Flow Contiguity Constraint}

For each region \(j\in\{1,\dots,p\}\), a unique root (source) area \(r(j)\) is predetermined (so that \(d_{r(j)}=j\)). We employ the scattered-seeding that enforces a minimum distance among the selected roots, ensuring that the subsequent regions remain spatially well-separated rather than being densely clustered. Then, for every area \(i\) assigned to region \(j\) the following holds:

\[
Inflow - Outflow = \]
\[\sum_{j\in N(i)} f_{ji} - \sum_{j\in N(i)} f_{ij} =
\begin{cases}
1, & \text{if } i\neq r(k),\\[1mm]
\displaystyle 1-\Bigl(\sum_{i\in I} \mathbf{1}\{d_i=k\}\Bigr) , & \text{if } i=r(k),
\end{cases}
\]
where \(\mathbf{1}\{d_i=j\}\) is an indicator function equal to 1 if \(d_i=j\) and 0 otherwise.
The first term represents the inflow, and the second term represents the outflow. For non-root nodes, the net inflow is 1; whereas for the root node, it must supply exactly enough flow to all other nodes in the region. In this way, by enforcing a balance between inflows and outflows, we ensure that each region forms a connected component, thereby achieving connectivity.

\begin{theorem}
If each region satisfies the flow conservation constraints, where the source satisfies and every other area satisfies, then the region must be connected.
\end{theorem}

\begin{proof}
Assume, for contradiction, that region is \emph{not} connected. That is, there exists a non-empty subset  of areas $S$ such that:

\begin{itemize}
\item  The set of areas $S$ forms a connected component that is \emph{disjoint} from the designated source node $s$ in the region source.
\item No selected edge (i.e., an edge with nonzero flow) connects any node in $S$ with any node outside of $S$ in the same region $R$, where $R$ represents all areas assigned to nodes in the region.
\end{itemize}

Because every area in satisfies, summing the flow conservation constraints over all areas in, we obtain:
\begin{equation}
\sum_{i\in S} (\text{Inflow} - \text{Outflow}) = |S|.
\end{equation}
However, since it is disconnected from the rest of the region, no flow enters or exits $S$, meaning the total across is zero: This contradicts the previous equation.
Thus, our assumption that this region is disconnected must be false. It follows that every area in a region is connected to the source if the flow constraint is met, ensuring spatial contiguity.
\end{proof}

Figure~\ref{fig:figure2} illustrates two scenarios:
(a) a connected Region~3, and (b) a disconnected Region~3.
In the connected case, all areas in Region~1 have $d_i = 1$, 
while areas $6$ and $9$ differ by $u_{69} = |d_{6} - d_{9}| = |3 - 1| = 2 \neq 0$, 
indicating they belong to different regions. 
Area~6 is designated as Region~3's source.
Figure~\ref{fig:connected} shows flow $f_{ij}$ values that satisfy the flow constraint (each non-root area has \text{Inflow} - \text{Outflow} = 1.).
In Figure~\ref{fig:connected}, area~6 is the source node, and $\text{Inflow} - \text{Outflow} = f_{76} - f_{67} = 1 - 3 = -2 = 1 - |Region~3|$; for area~7, $\text{Inflow} - \text{Outflow} = f_{67} + f_{57} - f_{76} - f_{75} = 3 + 0 - 1 - 1 = 1$; and for area~5, $\text{Inflow} - \text{Outflow} = f_{75} - f_{57} = 1 - 0 = 1$.
In this case, all areas in the region meet the flow contiguity constraint, so the region is connected.
In Figure~\ref{fig:disconnected}, the region is disconnected, and $\text{Inflow} - \text{Outflow} = f_{67} - f_{76} = 0 - 0 = 0$; for area~7, $\text{Inflow} - \text{Outflow} = f_{67} + f_{57} - f_{76} - f_{75} = 0 + 0 - 0 - 0 = 0$; and for area~5, $\text{Inflow} - \text{Outflow} = f_{75} - f_{57} = 0 - 0 = 0$.

\vspace{-4pt}
\subsection*{Objective Function}

A common objective is to minimize the overall inter-regional dissimilarity. For example, if \(w_{ij}\) is a measure of dissimilarity between areas \(i\) and \(j\), the objective function is
\[
\text{Minimize}\quad \sum_{k=1}^{p}\sum_{\substack{i,j\in I \\ i<j}} w_{ij}\,\mathbf{1}\{d_i=k\}\,\mathbf{1}\{d_j=k\}.
\]

The objective function is implemented in DQM by using the \texttt{{set\_quadratic\_case}} method to set biases for pairs of areas when they are assigned to the same region, aligning with the indicator function approach.

\subsection*{Complete DQM Formulation}

\[
\textbf{Decision Variables:}
\]
\[
d_i \in \{1,\ldots,p\}, \quad \forall\, i \in I
\]
\[
f_{ij} \ge 0, \quad \forall\,(i,j) \in E
\]
\[
u_{ij} \ge 0, \quad \forall\,(i,j) \in E
\]

\[
\textbf{Objective Function:}
\]
\[
\min \quad \sum_{j=1}^{p}\sum_{\substack{i,k\in I \\ i<k}} w_{ik}\,\mathbf{1}\{d_i=j\}\,\mathbf{1}\{d_k=j\}
\]

\[
\textbf{Subject to:}
\]
\[
u_{ij} \ge \lvert d_i - d_j\lvert, \quad \forall\,(i,j) \in E
\]

\[
f_{ij} \le M\,(1 - u_{ij}), \quad \forall\,(i,j) \in E
\]
\[
\sum_{j\in N(i)} f_{ji} - \sum_{j\in N(i)} f_{ij} = 1, \quad \forall\, i\in I \text{ with } i \neq r(d_i)
\]
\[
\sum_{j\in N(r(k))} f_{r(k)j} - \sum_{j\in N(r(i))} f_{j,r(k)} = \left(\sum_{i\in I} \mathbf{1}\{d_i=k\}\right) - 1. 
\]

\vspace{4pt}
\noindent In the above model:
\begin{itemize}
    \item The \(d_i\) are the only assignment variables (one per area).
    \item The variables \(f_{ij}\) and \(u_{ij}\) (for \((i,j)\in E\)) enforce that flow may occur only between areas in the same region and that a spanning-tree-like structure is formed within each region.
    \item The flow conservation constraints ensure that every non-root area receives one unit of flow and that the designated root supplies \(\Bigl(\sum_{i\in I}\mathbf{1}\{d_i=j\}\Bigr)-1\) units.
\end{itemize}

\section{QUBO Formulation}
After DQM formulation, the subsequent step involves its adaptation for quantum annealing, through QUBO representation. D-Wave systems, for instance, can directly transform smaller-scale DQM instances into the QUBO format. For larger, more complex models that exceed the direct capacity of the QPU, Dwave Leap Solver can decompose the large-scale DQM into multiple, smaller QUBO subproblems that can then be individually processed and their solutions potentially recombined.

\subsection*{Decision Variables}
The original assignment variable is 
\[
d_i \in \{1,\ldots,p\}, \quad \forall\, i \in I.
\]
For the assignment variable, we encode $d_i$  with binary variables:
\[
x_{ik} =
\begin{cases}
1, & \text{if } d_i = k,\\[1mm]
0, & \text{otherwise},
\end{cases}
\quad \forall\, i\in I,\; k=1,\ldots,p.
\]
To guarantee a unique assignment for each $i$, i.e., each area must only be assigned to one region, we require
\[
\sum_{k=1}^{p} x_{ik} = 1, \quad \forall\, i \in I,
\]
which is enforced via the penalty term
\[
\lambda_1 \left(\sum_{k=1}^{p} x_{ik} - 1\right)^2.
\]

\subsection*{Auxiliary Variables}
For modeling flows over network edges, we introduce flow variables
\[
f_{ij} \ge 0, \quad \forall\, (i,j) \in E.
\]
Since $f_{ij}$ is integer-valued, we encode it as a binary expansion:
\[
f_{ij} = \sum_{\ell=1}^{L} 2^{\ell-1}\, y_{ij}^{(\ell)}, \quad
y_{ij}^{(\ell)} \in \{0,1\},
\] 
with $L$ chosen such that $2^L-1\ge M$, where $M$ is an upper bound 
on the flow.

In addition, we introduce auxiliary variable $u_{ij,k}$ and penalty $P_{ij,k}(x,u)$ to keep the model formulation quadratic:
\[
u_{ij,k} \;=\; x_{ik}\,x_{jk}, 
\qquad
P_{ij,k}(x,u) \;=\; x_{ik}x_{jk}
  -2\,x_{ik}u_{ij,k}
  -2\,x_{jk}u_{ij,k}
  +3\,u_{ij,k}.
\]
To link the flows with the assignments, we require that 
\emph{if} $x_{ik}\neq x_{jk}$, then $f_{ij}$ must be 0. this gives the penalty term:
\begin{align*}
&\lambda_2\,f_{ij}\!\left(1-\sum_{k=1}^{p}x_{ik}x_{jk}\right)
   +\lambda_2\sum_{k=1}^{p}P_{ij,k}(x,u)\\[2pt]
&=\;\lambda_2\sum_{(i,j)\in E}
      \Bigl[
        f_{ij}\!\left(1-\sum_{k=1}^{p}u_{ij,k}\right)
        +\sum_{k=1}^{p}P_{ij,k}(x,u)
      \Bigr].
\end{align*}

\subsection*{Complete QUBO Formulation}
The original objective is to minimize
\[
\min \; \sum_{k=1}^{p} \sum_{i<j} w_{ij}\, x_{ik}\, x_{jk}.
\]
About flow conservation constraints. For 
each non-root node $i$, 
\[
\sum_{j\in N(i)} f_{ji} - \sum_{j\in N(i)} f_{ij} = 1, 
\quad \forall\, i \in I\setminus\{r(d_i)\},
\]
with an associated penalty term
\[
\lambda_3 \left(\sum_{k\in N(i)} f_{ki} - \sum_{k\in N(i)} f_{ik} - 1\right)^2.
\]
For each facility (with designated root $r(k)$), the conservation 
constraint is
\[
\sum_{j\in N(r(k))} f_{r(k)j} - \sum_{j\in N(r(k))} f_{j,r(k)} = 
\left(\sum_{i\in I} x_{ik}\right) - 1, \quad \forall\, k=1,\dots,p,
\]
and its penalty term is
\[
\lambda_4 \left(\sum_{j\in N(r(k))} f_{r(k)j} - 
\sum_{j\in N(r(k))} f_{j,r(k)} - \left(\sum_{i\in I} x_{ik} - 1\right)
\right)^2.
\]

Thus, the complete QUBO objective function is formulated as
\begin{align*}
\min_{x,y}\quad Q(x,y) =\; & \sum_{k=1}^{p} \sum_{i<j} w_{ij}\,
x_{ik}\, x_{jk} 
+ \lambda_1 \sum_{i\in I} \left(\sum_{k=1}^{p} x_{ik} - 1\right)^2 \\[1mm]
&+ \;\lambda_2\sum_{(i,j)\in E}
      \Bigl[
        f_{ij}\!\left(1-\sum_{k=1}^{p}u_{ij,k}\right)
        +\sum_{k=1}^{p}P_{ij,k}(x,u)
      \Bigr] \\[1mm]
&+ \lambda_3 \sum_{i\in I\setminus\{r(d_1),\dots,r(d_p)\}} \left(
\sum_{k\in N(i)} f_{ki} - \sum_{k\in N(i)} f_{ik} - 1\right)^2 \\[1mm]
&+ \lambda_4 \sum_{k=1}^{p} \left(
    \sum_{j\in N(r(k))} f_{r(k)j} - \sum_{j\in N(r(k))} f_{j,r(k)}
\right. \\[0.5mm]
&\hspace{3.5cm} \left.
    - \left(\sum_{i\in I} x_{ik} - 1\right)
\right)^2.
\end{align*}

Each penalty coefficient $\lambda_i$ is chosen sufficiently high to enforce large cost on constraint violations and ensure spatial contiguity.

\section{Model Scalability: A Hybrid Approach – The Seeding-based Spatial Regionalization}


Flow-based modeling of spatial contiguity typically requires a quadratic number of variables to model flow between all pairs of spatial areas.
Our proposed model optimizes this by representing only pairs of areas that are actual spatial neighbors (as shown in Figure~\ref{fig:figure2}) to exploit a lower degree of connectivity between spatial areas in reality and reduce the total number of variables and qubits needed.
However, given the existing limitations of quantum hardware, this could still be expensive for large-scale datasets, making it infeasible to directly deploy a model for the entire dataset on a QPU using either DQM or QUBO formulations. To address this, we introduce a seeding strategy that enables using our model to solve large-scale problems within a hybrid classical-quantum approach by dividing them into smaller subproblems. 

In the context of spatial regionalization, our use case in this paper, seeding-based approaches have shown the best scalability by far for large datasets~\cite{Aydin02112021, wei2021efficient, AlrashidLM23, Alrashid2022SMP, Liu2022PRUC, Kang2022EMP}, compared to graph partitioning approaches~\cite{Martins2006MST, Aydin2018SKATERCON, Guo2008} and much more scalable than MIP-based models~\cite{duque2011pregions, duque2012maxp, Kang2022EMP, Liu2022PRUC}.
Seeding goes through a common optimization framework that consists of different subproblems:
(1)~A growing phase that finds an initial solution through growing the required set of regions based on a set of seed areas.
(2)~A local optimization phase that optimizes the initial solution quality through heuristics to produce the final approximate solution.
For the problem of p-regions, the number of seed areas is p,
and each seed corresponds to one region. After seed selection, the
region grows by subsequently assigning neighboring areas. The
region-growing phase ensures spatial contiguity by restricting the
assignment to adjacent areas. More detailed steps are available in our supplementary material~\cite{chang2025quantummodelingspatialcontiguity}.

Finding \(p\) regions over a large dataset can be divided into:
(1)~seed selection (on QPU),
(2)~single-region growing, for \(p\) regions (on CPU),
(3)~combining the initial solution (on CPU),
(4)~determining movable areas (on CPU), and
(5)~optimizing the initial solution to get the final solution (on QPU).
We will limit the rest of our discussion to the quantum-driven steps for space limitations.

\textbf{Seed selection}.
A classic optimization for seed selection maximizes the minimum distance between seeds~\cite{Liu2022PRUC}, so seeds are spread enough to give the most room to find alternative solutions.
This optimization is a strong candidate for using quantum optimization by modeling the seed selection problem into a max-min dispersion problem~\cite{Yuki2024Disp}. Then we put high-quality seeds into region growing phase.


\textbf{The quantum optimization phase}.
This phase exploits quantum optimization to improve the objective function quality through exploring alternative solutions that are generated from the initial solution. Because we have limited the number of variables by only focusing on \emph{movable areas}, we can adapt to the limitation of small quantum hardware.
To generate alternative solutions, \emph{movable areas} that are located at each region's border are identified, so that they are not articulation areas that violate the spatial contiguity if reassigned.
We typically use Tarjan’s algorithm~\cite{Tarjan1972DFS} to identify articulation areas and filter them out.
Then, alternative solutions are combinations of the initial solution and the valid movable areas.
Quantum Annealing is powerful to explore many alternative solutions simultaneously and find an optimized one quickly.
\section{Conclusion}

This paper models the spatial contiguity constraint using quantum modeling to explore the potential quantum advantage in regionalization problems.
Specifically, we formulate a flow constraint in DQM to guarantee spatial contiguity, significantly reducing the required number of variables and effectively enabling quantum annealing to explore a broader solution space.  
Our ongoing work will empirically validate our model and extend our hybrid approach to more spatial optimization problems. 

\bibliographystyle{ACM-Reference-Format}
\bibliography{sample-base}

\appendix

\end{document}